\documentclass[onecolumn]{IEEEtran}
\usepackage{amsmath,amssymb,amsfonts,bm,amsthm}
\usepackage{graphics,graphicx,epsfig,subfigure}
\usepackage{multirow,multicol}
\usepackage{algorithm,algpseudocode}

\newtheorem{theorem}{Theorem~}%[section]

\newtheorem{remark}{Remark~}
\newtheorem{lemma}{Lemma~}

\newtheorem{property}{Problem~}

\ifCLASSINFOpdf
\else
\fi
\begin{document}
\title{Linear Quadratic Synchronization of Multi-Agent Systems: A Distributed Optimization Approach}

\author{Qishao~Wang,
        Zhisheng~Duan$^\star$,~\IEEEmembership{Member,~IEEE,}
        Jingyao~Wang,
        and~Guanrong~Chen,~\IEEEmembership{Fellow,~IEEE}  % stops a space
\thanks{Q.~Wang and Z.~Duan are with the State Key Laboratory for Turbulence and Complex Systems, Department of Mechanics and Engineering Science, College of Engineering, Peking University, Beijing 100871, P.~R.~China (email: duanzs@pku.edu.cn).}% <-this % stops a space
\thanks{J.~Wang are with the Department of Automation, Xiamen University, Xiamen 361000, P.~R.~China.}% <-this % stops a space
\thanks{G.~Chen is with the Department of Electronic Engineering, City University of Hong Kong, Hong Kong SAR, P.~R.~China.}}



\markboth{Journal of \LaTeX\ Class Files,~Vol.~14, No.~8, August~2015}%
{Shell \MakeLowercase{\textit{et al.}}: Bare Demo of IEEEtran.cls for IEEE Journals}



\maketitle

\begin{abstract}
The distributed optimal synchronization problem with linear quadratic cost is solved in this paper for multi-agent systems with an undirected communication topology. For the first time, the optimal synchronization problem is formulated as a distributed optimization problem with a linear quadratic cost functional that integrates quadratic synchronization errors and quadratic input signals subject to agent dynamics and synchronization constraints. By introducing auxiliary synchronization state variables and combining the distributed synchronization method with the alternating direction method of multiplier (ADMM), a new distributed control protocol is designed for solving the distributed optimization problem. With this construction, the optimal synchronization control problem is separated into several independent subproblems: a synchronization optimization, an input minimization and a dual optimization. These subproblems are then solved by distributed numerical algorithms based on the Lyapunov method and dynamic programming. Numerical examples for both homogeneous and heterogeneous multi-agent systems are given to demonstrate the effectiveness of the proposed method.
%This dual problem is then distributedly solved, based on the Alternating Direction Multiplier Method (ADMM)
\end{abstract}

\begin{IEEEkeywords}
Distributed Optimization; Synchronization; Heterogeneous Systems; Control System.
\end{IEEEkeywords}

% \fi
\IEEEpeerreviewmaketitle

\section{Introduction}
The synchronization control problem for multi-agent systems has attracted considerable attention due to its various applications to many important tasks \cite{OlfatiSR2007,RenW2007}, such as formation flying of unmanned air vehicles, spacecraft attitude cooperative control, distributed sensor configuration and information flow control. A great number of existing works on multi-agent systems mainly focus on the synchronization problem on networks with various topologies \cite{LiZK2010,RenW2005,WenGH2017TSMCS}, communication constraints \cite{WenGH2016,YuWW2013}, complex dynamics \cite{DuHB2017}, final state restrictions \cite{WenGH2017TSMCS2}, robustness \cite{LiZK2017} and so on. In practice, it is desirable to improve some control performances such as convergence rate and control energy cost while achieving synchronization, which is typically the goal of distributed optimization.

The distributed optimization problem for multi-agent systems has been widely investigated recently. Some earlier works are presented in \cite{ShiG2013,LinP2014}, where the dynamics of agents are described by integrators. Combining synchronization control methods with optimization techniques, the optimal synchronization problem was solved for double-integrator dynamics \cite{QiuZR2016} and then extended to Euler-Lagrangian systems \cite{ZhangYQ2017}, where the final synchronization state is required to minimize a global cost functional. For general linear dynamics \cite{ZhaoY2017}, cooperative optimization is achieved through local interactions by implementing edge- or node-based adaptive algorithms.
%In the distributed optimization algorithms mentioned above, the objective functional is formed by a sum of convex local objective functionals depending on the final synchronization value.
To optimize the transient response of the synchronization process, the objective functional is reformulated to be an integral of synchronization error over time in \cite{WangJY2013,JiaoQ2016}. $H_\infty$ and $H_2$ control protocols are proposed in \cite{WangJY2013} for multi-agent systems to achieve synchronization synthesised with desired transient performance. $\mathcal{L}_2$-gain output-feedback synchronization problems for both homogeneous and heterogeneous multi-agent systems are addressed in \cite{JiaoQ2016}, to achieve synchronization and meanwhile limit the $\mathcal{L}_2$-gain of the synchronization error. When combining the transient response of synchronization together with the control energy cost, the distributed optimization problem for linear multi-agent systems becomes the distributed linear quadratic synchronization problem, where the objective functional integrates the quadratic synchronization error and quadratic input signals.

One case of distributed linear quadratic synchronization is the linear quadratic regulator (LQR), where all the agents are required to be stabilized with a quadratic cost functional minimized \cite{CaoY2010,KempkerPL2014,NguyenDH2014}. The LQR optimal synchronization problem is studied in \cite{CaoY2010}, where the communication topology corresponds to a complete graph. The overall LQR control problem is separated into independent local subproblems for coordinated linear systems thereby deriving a lower-order distributed numerical algorithm in \cite{KempkerPL2014}. For an undirected communication topology, in \cite{NguyenDH2014} a distributed stabilizing control approach is taken to minimize the LQR performance index, where the involving weighted matrices have to be properly chosen. %Naturally, if the agents are only required to be synchronized the objective functional is able to get be minimized \cite{MontenbruckJM2015,ZengS2017,ModaresH2017}.
Based on the algebraic Riccati equation, optimal control protocols with diffusive couplings are presented in \cite{MontenbruckJM2015} for linear synchronization problems with quadratic cost and the results are extended to a static output feedback scenario in \cite{ZengS2017}. For the leader-follower synchronization problem \cite{ModaresH2017}, the Hamilton-Jacobi-Bellman equation is utilized to find an optimal control protocol based on distributed estimation of the leader state for each follower agent. It should be noted, however, that despite the considerable advances on distributed optimization, the problem of designing distributed optimal synchronization algorithms with general linear quadratic cost functionals remains a challenge.
%to the best of our knowledge, there are scare results in the open literature on the distributed design algorithm for the synchronization problem with general linear quadratic objective functionals under undirected communication topology.

Motivated by the above observations, a distributed optimization algorithm is proposed in this paper to achieve optimal synchronization minimizing a linear quadratic cost for multi-agent systems with an undirected communication topology. By introducing some auxiliary synchronization state variables, the optimal synchronization problem is formulated as a distributed optimization problem subject to reguired agent dynamics and synchronization constraints with a linear quadratic cost functional that integrates quadratic synchronization error and quadratic input signals. A new distributed control protocol design framework is proposed by combining the distributed synchronization method with the alternating direction method of multiplier (ADMM). With this construction, the optimal synchronization control problem is separated to several independent subproblems: a synchronization optimization, an input minimization and a dual optimization. Then, a distributed numerical algorithm corresponding to each subproblem is designed based on the Lyapunov method and dynamic programming. Comparing with the literature on distributed optimization control, the contributions of this paper are three-fold, as summarized below:
\begin{enumerate}[]
\item A new distributed control protocol design is proposed by combining the distributed synchronization method with the ADMM for the linear quadratic synchronization control problem. For the first time, a variant of the generalized ADMM algorithm is applied to separate the optimal synchronization control problem to several independent subproblems that can be solved in a distributed way. A further convergence analysis shows that the control sequence generated by the proposed algorithm converges to the optimal solution of the linear quadratic synchronization control problem. This new framework is very desirable for distributed control protocol design since the communication topology and the agent dynamics are successfully separated, making the design and analysis much easier.
   % However, this application is not trivial, and indeed, it requires a rather elaborate analysis. In the linear quadratic synchronization control problem, the objective function is not separable across the variables because of the quadratic term of synchronization error, which is inconsistent with the standard ADMM formulation \cite{BoydS2010}. Therefore, a variant of ADMM algorithm suitable for the linear quadratic synchronization control problem is proposed and the convergence is analyzed in this paper inspired by \cite{GaoX2017}.
\item The synchronization control problem for multi-agent systems with linear quadratic cost is solved by a single-agent-level algorithm. As indicated in \cite{MontenbruckJM2015}, the quadratic term of the Laplacian matrix appears in the objective functional and in the Riccati equation, which brings more difficulties in order reduction. In this paper, the optimal synchronization control problem is divided into synchronization and optimal control by the ADMM technique. In the synchronization step, the optimal synchronization state for each iteration is solved by differential equations using local information. Then, optimal control input can be designed individually for each agent in the optimal control step with the synchronization state fixed. Therefore, the design algorithm for optimal control has the same order as each agent in both steps. Moreover, the order reduction does not introduce additional constraints on the communication topology or the weighted matrices in the cost functional.
\item The distributed numerical algorithm is valid for both homogenous and heterogenous linear systems with eigenvalues either inside or on the unit circle, or for the eigenvalues outside the unit circle respectively. By an application of the ADMM technique, the topology issue is removed from the optimal control input design step so that the design algorithm can be easily applied to general heterogenous linear systems. On the other hand, the dynamic programming scheme used in solving the optimal control input ensures a stable final synchronization state for both stable and unstable dynamics.
\end{enumerate}

The rest of this paper is organized as follows. In Section \ref{s:Preliminaries and Problem Formulation}, some preliminaries and the formulation of the optimal synchronization problem with linear quadratic cost are presented. A variant of the generalized ADMM algorithm and its convergence analysis for synchronization control in a centralized manner are presented in Section \ref{s:Distributed Linear Quadratic Synchronization Control Problem Design}. Section \ref{s:subproblemalgorithm} develops distributed algorithms for the synchronization, the control design and the overall optimal synchronization problem, respectively. The performances of the proposed algorithms are illustrated by numerical examples in Section \ref{s:Simulations}, with conclusions given in Section \ref{s:conclusion}.

The notations used in this paper are as follows. The set of $n$-dimensional real vectors and $m\times n$ real matrices are indicated by $\mathbb{R}^n$ and $\mathbb{R}^{m\times n}$, $\otimes$ denotes the Kronecker product of matrices, and $\|\cdot\|$ denotes the Euclidean norm of the corresponding vector and matrix. For $x_i\in\mathcal{R}^{n_i},~A_i\in\mathcal{R}^{m_i\times n_i}~i=1,\cdots,m$, define $col\{x_1,\cdots,x_m\}\triangleq [x_1^T,\cdots,x_m^T]^T$ and $diag\{A_1,\cdots,A_m\}$ be a block diagonal matrix.

\section{Preliminaries and Problem Formulation}\label{s:Preliminaries and Problem Formulation}
%\subsection{Kinematics and Dynamics of Flexible Spacecraft Attitude}\label{ss:Kinematics and Dynamics of Flexible Spacecraft Attitude}
Consider a network of $N$ heterogeneous agents with discrete-time linear dynamics in the following form
\begin{equation}\label{e:systemoriginal}
x_i(k+1)=A_ix_i(k)+B_iu_i(k),~i\in\{1,2,\cdots,N\},
\end{equation}
where $x_i\in\mathbb{R}^n$ is the state of the $i$-th agent, $u_i\in\mathbb{R}^m$ is its control input and $A_i\in \mathbb{R}^{n\times n},~B_i\in \mathbb{R}^{n\times m}$ are constant matrices.

The agents are assumed to exchange information through a communication network described by an undirected and connected graph $\mathcal{G}=(\mathcal{V},\mathcal{E})$, with $\mathcal{V}=\{v_1,~v_2,\cdots,~v_N\}$ being the set of nodes and $\mathcal{E}\subset\mathcal{V}\times\mathcal{V}$ being the set of edges. In the graph $\mathcal{G}$, $(v_i,~v_j)\in\mathcal{E}$ means that the $i$-th agent can exchange information with the $j$-th agent. The weighted adjacency matrix of the graph $\mathcal{G}$ is defined as $\mathcal{A}=(a_{ij})\in\mathbf{R}^{N\times N}$, where $a_{ii}=0$, and $a_{ij}=a_{ji}>0$ if $(v_i,~v_j)\in\mathcal{E}$. The Laplacian matrix of $\mathcal{G}=(\mathcal{V},\mathcal{E})$ is denoted by $\mathcal{L}=(l_{ij})\in\mathbf{R}^{N\times N}$, where $l_{ii}=\sum_{j=1}^N a_{ij}$,~$l_{ij}=-a_{ij}$ for $i\neq j$. And $\mathcal{V}_i=\{j\in\mathcal{V}:\{i,j\}\in\mathcal{E}\}$ denotes the neighborhood set of $i$.

The first problem considered in this paper is to find controllers $u_i$ to guarantee the synchronization of all agents, i.e.,
\begin{equation}\label{e:controlgoal}
\lim_{k\rightarrow \mathcal{N}}\|x_i(k)-x_j(k)\|=0,~\forall i,j\in\{1,2,...,N\},
\end{equation}
where $\mathcal{N}$ denotes the total (finite) steps needed to achieve synchronization.
Denote the finite synchronization state as $z_i,~i\in\{1,2,...,N\}$. Then, the synchronization condition (\ref{e:controlgoal}) can be rewritten as
\begin{equation}\label{e:controlgoal2}
\begin{aligned}
&\lim_{k\rightarrow \mathcal{N}} x_i(k)=z_i,~i\in\{1,2,...,N\},\\
&(\mathcal{L}\otimes I_n)Z=0,
\end{aligned}
\end{equation}
where $Z\triangleq col\{z_1,z_2,...,z_N\}$.
Define the synchronization error vector of the network as
\begin{equation}\label{e:synerror}
e_i(k)=x_i(k)-z_i,~~i\in\{1,2,...,N\}.
\end{equation}
Let the control input sequence be $u_i \triangleq col\{u_i(0),u_i(1),\cdots,u_i(\mathcal{N}-1)\}$, $U\triangleq col\{u_1,u_2,\cdots,u_N\}$, and the cost functional
\begin{equation}\label{e:cost}
\begin{aligned}
J(U,Z)=&\sum_{i=1}^{N}\left\{\sum_{k=0}^{\mathcal{N}-1}\left[e_i^T(k)Q_ie_i(k)+u_i^T(k)R_iu_i(k)\right]+e_i^T(\mathcal{N})Q_{i\mathcal{N}}e_i(\mathcal{N})\right\},
\end{aligned}
\end{equation}
for some $Q_{i\mathcal{N}},Q_i\in\mathbb{R}^{n\times n},R_i\in \mathbb{R}^{m\times m}$ with $Q_{i\mathcal{N}}\geq0,Q_i\geq0,R_i>0$. Physically, this quadratic cost functional is composed of the energies of the error signal and of the input signal. It can be used as a performance index to quantify the swiftness, vibration and energy consumption of the network synchronization. Consequently, the second problem is to design a control sequence $U^\star$ that minimizes (\ref{e:cost}) subject to (\ref{e:systemoriginal}), which implicitly achieves synchronization as $\mathcal{N}$ becomes large enough.
\begin{property}\label{p:LQR}
Combining the two problems mentioned above, the linear quadratic synchronization control problem can be expressed as
\begin{equation}\label{e:optimizationproblem}
\begin{aligned}
\min \limits_{U,Z}~&J(U,Z)\\
s.t.~&x_i(k+1)=A_ix_i(k)+B_iu_i(k),~i\in\{1,2,...,N\}\\
&(\mathcal{L}\otimes I_n)Z=0.
\end{aligned}
\end{equation}
\end{property}

\begin{remark}\label{LQRsignificance}
In the cost functional $J(u_i,z_i,x_0)$, the terms $e_i^T(k)Q_ie_i(k)$ and $e_i^T(\mathcal{N})Q_{i\mathcal{N}}e_i(\mathcal{N})$ are introduced to improve the synchronization rate and the final synchronization precision respectively. The weighted matrices $Q_i$ and $Q_{i\mathcal{N}}$ are set to be positive semi-definite so that the familiar output synchronization can be regarded as a special case of Problem \ref{p:LQR} here. For example, if the output of agent $i$ is described by $y_i(k)=C_ix_i(k)$, the synchronization error becomes $e_{io}(k)=C_ix_i(k)-C_iz_i=C_ie_i(k)$, where $C_i$ may not be of full row rank. Thus, the output synchronization error term in the cost functional can be selected as $e_i(k)C_i^TC_ie_i(k)$. In this case, to achieve output synchronization, matrices $Q_{i} $ and $Q_{i\mathcal{N}}$ can be selected as $Q_{i}=Q_{i\mathcal{N}}=C_i^TC_i\ge 0$. Moreover, $u_i^T(k)R_iu_i(k)$ acts as a control penalty on the control input power. In fact, without this term the amplitude of the control input will go to infinity since maintaining smaller synchronization error requires larger control input. Thus, the weighted matrix $R_i$ should be positive definite to restrict all the components of the control input vector within a reasonable range. In a real design, the selection of $Q_i, Q_{i\mathcal{N}}$ and $R_i$ implies  a tradeoff among synchronization rate, final synchronization error and control energy.
\end{remark}

\section{A Centralized Algorithm for Synchronization Control}\label{s:Distributed Linear Quadratic Synchronization Control Problem Design}
In this section, consider the optimal linear quadratic synchronization control problem (\ref{e:optimizationproblem}).
Using the method of multipliers, the augmented Lagrangian is first formulated as follows:
\begin{equation}\label{e:Lagrangian}
\begin{aligned}
L_\rho(U,Z,\Lambda)=&\sum_{i=1}^{N}\left\{e_i^T(\mathcal{N})Q_{i\mathcal{N}}e_i(\mathcal{N})+\sum_{k=0}^{\mathcal{N}-1}\left[e_i(k)^TQ_ie_i(k)+u_i(k)^TR_iu_i(k)\right]\right\}\\
&+\Lambda^T(\mathcal{L}\otimes I_n)Z+\frac{\rho}{2}Z^T(\mathcal{L}\otimes I_n)Z,
\end{aligned}
\end{equation}
where $\Lambda\triangleq col\{\lambda_1,\lambda_2,...,\lambda_n\}$ is the Lagrangian multipliers and $\rho>0$ is the augmented Lagrangian parameter. Then, a variant of the ADMM algorithm proposed in \cite{BoydS2010,GaoX2017} can be applied, which consists of the iterations (\ref{se:admmC}).
\begin{algorithm}[htb]
\caption{Centralized Linear Quadratic Synchronization Control Algorithm}
\label{a:ADMM}
Initialize $U^0, Z^0$ and $\Lambda^0$.  For $q=0,1,...$, until convergent:
%\begin{subequations}\label{se:admm}
%\begin{equation}\label{e:admmu}
%u_i^{q+1}(k)=\arg \min \limits_{u_i(k)} L_\rho(u_i,z_i^q,x_0,\Lambda^q)+\frac{1}{2}(u_i-u_i^q)^TG_i(u_i-u_i^q),
%\end{equation}
%\begin{equation}\label{e:admmz}
%z_i^{q+1}=\arg \min \limits_{z_i} L_\rho(u_i^{q+1},z_i,x_0,\Lambda^q)+\frac{1}{2}(z_i-z_i^q)^TH_i(z_i-z_i^q),
%\end{equation}
%\begin{equation}\label{e:admml}
%\lambda^{q+1}_i=\lambda^q_i+\rho\sum_{j\in \mathbb{N}_i}(z_i^{q+1}-z_j^{q+1}).
%\end{equation}
%\begin{equation}\label{e:admmlambda}
%\Lambda^{q+1}=\Lambda^q+\rho(\mathcal{L}\otimes I_n)Z^{q+1}.
%\end{equation}
%\end{subequations}
%\begin{subequations}\label{se:admm}
%\begin{align}
%\label{e:admmz}
%z_i^{q+1}&=\arg \min \limits_{z_i}\{L_\rho(u_i^{q},z_i,x_0,\Lambda^q)+\frac{1}{2}(z_i-z_i^q)^TG_i(z_i-z_i^q)\},\\
%\label{e:admmu}
%u_i^{q+1}&=\arg \min \limits_{u_i} \{L_\rho(u_i,z_i^{q+1},x_0,\Lambda^q)+\frac{1}{2}(u_i-u_i^q)^T(I_\mathcal{N}\otimes H_i)(u_i-u_i^q)\},\\
%\label{e:admml}
%\lambda^{q+1}_i&=\lambda^q_i+\rho z_i^{q+1}.
%%\label{e:admmlambda}
%%\Lambda^{q+1}&=\Lambda^q+\rho(\mathcal{L}\otimes I_n)Z^{q+1}.
%\end{align}
%\end{subequations}
\begin{subequations}\label{se:admmC}
\begin{equation}\label{e:admmZ}
\begin{aligned}
Z^{q+1}=&\arg \min \limits_{Z}\left\{L_\rho(U^{q},Z,\Lambda^q)+\textstyle\frac{1}{2}\displaystyle(Z-Z^q)^TG(Z-Z^q)\right\},
\end{aligned}
\end{equation}
\begin{equation}\label{e:admmU}
\begin{aligned}
U^{q+1}=&\arg \min\limits_{U}\left\{L_\rho(U,Z^{q+1},\Lambda^q)+\textstyle\frac{1}{2}\displaystyle(U-U^q)^TH(U-U^q)\right\},
\end{aligned}
\end{equation}
\begin{equation}\label{e:admmL}
\lambda^{q+1}_i=\lambda^q_i+\rho z_i^{q+1},~~i\in\{1,2,...,N\},
%\label{e:admmlambda}
%\Lambda^{q+1}&=\Lambda^q+\rho(\mathcal{L}\otimes I_n)Z^{q+1}.
\end{equation}
\end{subequations}
where $G\triangleq diag\{G_1,G_2,\cdots,G_N\}$ and $H\triangleq diag\{I_\mathcal{N}\otimes H_1,I_\mathcal{N}\otimes H_2,\cdots,I_\mathcal{N}\otimes H_N\}$.
\end{algorithm}

In Algorithm \ref{a:ADMM}, matrices $G_i$ and $H_i$ are chosen positive matrices. This algorithm divides the linear quadratic synchronization control problem (\ref{e:optimizationproblem}) into  a $Z$-minimization step (\ref{e:admmZ}), a $U$-minimization step (\ref{e:admmU}) and a dual variable update step (\ref{e:admmL}), which separates the node dynamics and the communication topology. Therefore, step (\ref{e:admmU}) can be regarded as a linear quadratic tracking problem with respect to individual subsystems and steps (\ref{e:admmZ}, \ref{e:admmL}) are used to achieve synchronization on the communication topology. In fact, Algorithm \ref{a:ADMM} is a variant of the generalized ADMM proposed in \cite{GaoX2017}. Then, the convergence analysis of Algorithm \ref{a:ADMM} is presented in the following Theorem whose proof can be found in Appendix.

%\begin{lemma}\label{l:optimalsaddle}
%For any convex function $f$ on $\mathbb{R}^m$ which is continuously differentiable with gradient $\nabla f$ which satisfies the Lipschitz continuous condition as follows
%\begin{equation}\label{Lipschitzf}
%\|\nabla f(x)-\nabla f(y)\|\le L_f\|x-y\|,~~~~~\forall x,y\in \mathbb{R}^m,
%\end{equation}
%we have
%\begin{equation}\label{e:convexeq}
%f(x)\le f(y)+\nabla f(z)^T(x-y)+\frac{L_f}{2}\|x-z\|^2,~~~~~\forall x,y,z\in \mathbb{R}^m.
%\end{equation}
%\end{lemma}

\begin{theorem}\label{t:convergence}
Suppose that $Q_{i\mathcal{N}}\geq0,Q_i\geq0,R_i>0$ and the final time step $\mathcal{N}$ is finite. Then, the sequence $\{U^q,Z^q\}$ generated by Algorithm \ref{a:ADMM} converges to an optimal solution if the following conditions are satisfied:
\begin{equation}\label{e:convergencecon}
G_i>0, H_i>\left({L_\delta}+\frac{L_\delta^2}{2\sigma_{min}\{R_i\}}\right)I_m,
\end{equation}
where $L_\delta$ is the Lispschitz constant for the gradient of the cost functional.
\end{theorem}

\begin{remark}\label{r:referadmm}
Theorem \ref{t:convergence} extends the existing results on the ADMM algorithm to deal with the distributed linear quadratic synchronization control problem. Comparing with the existing studies of distributed optimization control \cite{QiuZR2016}, the objective functional here is not necessarily separable across variables, i.e., the coupling functional $J_1(U,Z)$ appears in the cost functional. The objective becomes nonseparable because not only the final synchronization state but also the time cumulation of the synchronization error and control energy are considered here. This nonseparable objective functional makes it hard to directly apply the classical ADMM technique \cite{BoydS2010}, therefore its variant is proposed as the new Algorithm \ref{a:ADMM}. It is also worth noticing that the method leading to Theorem \ref{t:convergence} is, in essence, consistent with the generalized ADMM method proposed in \cite{GaoX2017}, where the convex optimization problem with a nonseparable objective functional is studied.
\end{remark}

\section{Distributed Synchronization Control}\label{s:subproblemalgorithm}
Based on the convergence result presented in Theorem \ref{t:convergence}, the linear quadratic synchronization control problem (\ref{e:optimizationproblem}) is successfully divided into a Z-minimization step (\ref{e:admmZ}) and a U-minimization step (\ref{e:admmU}) in Algorithm \ref{a:ADMM}, which however is still centralized. In this section, distributed algorithms for steps (\ref{e:admmZ}) and (\ref{e:admmU}) are derived respectively.

\begin{theorem}\label{t:step1}
If the communication topology is undirected and connected, then the optimal solution of (\ref{e:admmZ}) can be obtained at the equilibrium point of
\begin{equation}\label{e:step1}
\begin{aligned}
\dot z_i=&-\left(2\mathcal{N}Q_i+2Q_{i\mathcal{N}}+G_i\right)z_i-\sum_{j\in\mathcal{V}_i}\left[\rho(z_i-z_j)+(\lambda_i^q-\lambda_j^q)\right]+2Q_{i\mathcal{N}}x_i^q(\mathcal{N})+2\sum_{k=0}^{\mathcal{N}-1}Q_ix_i^q(k)+G_iz_i^q.
\end{aligned}
\end{equation}
\end{theorem}
\begin{proof}
First of all, rewrite (\ref{e:step1}) in a compact form:
\begin{equation}\label{e:step1compact}
\begin{aligned}
\dot Z=&-\left[2\mathcal{N}Q+2Q_\mathcal{N}+G+\rho(\mathcal{L}\otimes I_n)\right]Z-(\mathcal{L}\otimes I_n)\Lambda^q+2\sum_{k=0}^{\mathcal{N}-1}QX^q(k)+2Q_\mathcal{N}X^q(\mathcal{N})+GZ^q,
\end{aligned}
\end{equation}
where $X^q(k)=col\{x_1^q(k),...x_N^q(k)\}$, $Q=diag\{Q_1,...,Q_N\}$ and $Q_\mathcal{N}=diag\{Q_{1\mathcal{N}},...,Q_{N\mathcal{N}}\}$.
Consider the Lyapunov function $V_Z=\frac{1}{2}Z^TZ$, which has the time derivative
\begin{equation}\label{e:dV_Z}
\dot V_Z=-Z^T\left[2\mathcal{N}Q+2Q_\mathcal{N}+G+\rho(\mathcal{L}\otimes I_n)\right]Z,
\end{equation}
and it is negative definite because $Q\ge0,Q_\mathcal{N}\ge0,(\mathcal{L}\otimes I_n)\ge0,G>0$. Consequently, the solution of differential equation (\ref{e:step1compact}) will converge to its equilibrium point that satisfies the KKT condition \cite{Boyd2004} of (\ref{e:admmZ}):
\begin{equation}\label{e:equilibrium}
\begin{aligned}
0=&\left[2\mathcal{N}Q+2Q_\mathcal{N}+G+\rho(\mathcal{L}\otimes I_n)\right]Z+(\mathcal{L}\otimes I_n)\Lambda^q-2\sum_{k=0}^{\mathcal{N}-1}QX^q(k)-2Q_\mathcal{N}X^q(\mathcal{N})-GZ^q.
\end{aligned}
\end{equation}
In conclusion, the solution of algorithm (\ref{e:step1}) will converge to the optimal solution of (\ref{e:admmZ}) since $2\mathcal{N}Q+2Q_\mathcal{N}+G+\rho(\mathcal{L}\otimes I_n)$ is positive definite.
\end{proof}

The following theorem presents the optimal controller for each agent individually to solve the linear quadratic synchronization control problem.

\begin{theorem}\label{t:step2}
Given $z_i^{q+1}, u_i^q,~i\in\{1,..,N\}$, the cost functional $L_\rho(U,Z^{q+1},\lambda^q)$ in (\ref{e:admmU}) is minimized by, for each step $k=\mathcal{N}-1,\mathcal{N}-2,...,0$, the control input
\begin{equation}\label{e:controller}
\begin{aligned}
u_i^\star(k)=&-U_i(k)\left[V_i(k)x_i(k)-W_i(k)\right],i\in\{1,2,...,N\},
\end{aligned}
\end{equation}
where
\begin{equation}\label{e:step1iteration}
\begin{aligned}
U_i&(k)=(R_i+H_i+B_i^TS_{i1}(k+1)B_i)^{-1},\\
V_i&(k)=B_i^TS_{i1}(k+1)A_i,\\
W_i&(k)=-\frac{1}{2}B_i^TS_{i2}^T(k+1)+H_iu_i^q(k),\\
S_{i1}&(k)=Q_i+V_i^T(k)U_i^T(k)(R_i+H_i)U_i(k)V_i(k)+\left[A_i-B_iU_i(k)V_i(k)\right]^TS_{i1}(k+1)\left[A_i-B_iU_i(k)V_i(k)\right],\\
S_{i2}&(k)=2W_i^T(k)U_i^T(k)\left[B_i^TS_{i1}(k+1)A_i-(R_i+H_i)U_i(k)V_i(k)-B_i^TS_{i1}(k+1)B_iU_i(k)V_i(k)\right]\\
    &+S_{i2}(k+1)\left[A_i-B_iU_i(k)V_i(k)\right]-2(z_i^{q+1})^TQ_i+2u_i^q(k)^TH_iU_i(k)V_i(k),\\
S_{i3}&(k)=(z_i^{q+1})^TQ_iz_i^{q+1}+W_i^T(k)U_i^T(k)\left[R_i+H_i+B_i^TS_{i1}(k+1)B_i\right]U_i(k)W_i(k)\\
    &+S_{i2}(k+1)B_iU_i(k)W_i(k)+S_{i3}(k+1)+u_i^q(k)^TH_i\left[u_i^q(k)-2U_i(k)W_i(k)\right],\\
S_{i1}&(\mathcal{N})=Q_{i\mathcal{N}},~S_{i2}(\mathcal{N})=-2(z_i^{q+1})^TQ_{i\mathcal{N}},\\
S_{i3}&(\mathcal{N})=(z_i^{q+1})^TQ_{i\mathcal{N}}z_i^{q+1}.
\end{aligned}
\end{equation}
The optimal objective value is given by
\begin{equation}\label{e:optimalvalue}
L_\rho^{\star q}=\sum_{i=1}^{N}L_i^\star(0),
\end{equation}
where $L_i^\star(0)=x_i^T(0)S_{i1}(0)x_i(0)+S_{i2}(0)x_i(0)+S_{i3}(0)$ and $x_i(0)$ is the initial state of the $i$-th agent, $i\in\{1,2,...,N\}$.
%Let
%\begin{equation}\label{e:step1initial}
%S_{i1}(\mathcal{N})=Q_{i\mathcal{N}},~S_{i2}(\mathcal{N})=-2(z_i^{q+1})^TQ_{i\mathcal{N}},~S_{i3}(\mathcal{N})=(z_i^{q+1})^TQ_{i\mathcal{N}}z_i^{q+1}.
%\end{equation}
%Then for $k=\mathcal{N}-1,\mathcal{N}-2,...,0$, the iterative algorithm is proposed as follows
%\begin{equation}\label{e:step1iteration}
%\begin{aligned}
%U_i(k)=&(R_i+H_i+B_i^TS_{i1}(k+1)B_i)^{-1},\\
%V_i(k)=&B_i^TS_{i1}(k+1)A_i,\\
%W_i(k)=&-\frac{1}{2}B_i^TS_{i2}^T(k+1)+H_iu_i^q(k),\\
%S_{i1}(k)=&Q_i+\left[A_i-B_iU_i(k)V_i(k)\right]^TS_{i1}(k+1)\left[A_i-B_iU_i(k)V_i(k)\right]\\
%    &+V_i^T(k)U_i^T(k)(R_i+H_i)U_i(k)V_i(k),\\
%S_{i2}(k)=&2W_i^T(k)U_i^T(k)\left[B_i^TS_{i1}(k+1)A_i-B_i^TS_{i1}(k+1)B_iU_i(k)V_i(k)-(R_i+H_i)U_i(k)V_i(k)\right]\\
%    &+S_{i2}(k+1)\left[A_i-B_iU_i(k)V_i(k)\right]-2(z_i^{q+1})^TQ_i+2u_i^q(k)^TH_iU_i(k)V_i(k),\\
%S_{i3}(k)=&(z_i^{q+1})^TQ_iz_i^{q+1}+W_i^T(k)U_i^T(k)\left[R_i+H_i+B_i^TS_{i1}(k+1)B_i\right]U_i(k)W_i(k)\\
%    &+S_{i2}(k+1)B_iU_i(k)W_i(k)+S_{i3}(k+1)+u_i^q(k)^TH_i\left[u_i^q(k)-2U_i(k)W_i(k)\right],\\
%u_i^\star(k)=&-U_i(k)\left[V_i(k)x_i(k)-W_i(k)\right],\\
%L_i^\star(k)=&x_i^T(k)S_{i1}(k)x_i(k)+S_{i2}(k)x_i(k)+S_{i3}(k),
%\end{aligned}
%\end{equation}
%where $x_i(0)$ is the initial state of each agent. Then the control input $u_i^\star(k),k=0,1,...,\mathcal{N}-1$ minimizes the cost functional $L_\rho(u_i,z_i^{q+1},x_0,\Lambda^q)$ in (\ref{e:admmU}) subject to (\ref{e:systemoriginal}). The optimal objective value can be calculated by
%\begin{equation}\label{e:optimalvalue}
%L_\rho^{\star q}=\sum_{i=1}^{N}L_i^\star(0).
%\end{equation}
\end{theorem}
\begin{proof}
Mathematical induction and dynamic programming are used in this proof. Fist, (\ref{e:step1iteration}) is verified for $k=\mathcal{N}-1$. According to the optimization principle \cite{BellmanR1972}, the optimal control input $u_i^\star(\mathcal{N}-1)$ must satisfy
\begin{equation}\label{e:optlast}
\begin{aligned}
u_i^\star(\mathcal{N}-1)=\arg \min \limits_{u_i(\mathcal{N}-1)} J_i(\mathcal{N}-1),
\end{aligned}
\end{equation}
where
\begin{equation}\label{e:Jn1}
\begin{aligned}
J_i(\mathcal{N}&-1)=(x_i(\mathcal{N})-z_i^{q+1})^TQ_{i\mathcal{N}}(x_i(\mathcal{N})-z_i^{q+1})+(x_i(\mathcal{N}-1)-z_i^{q+1})^TQ_{i}(x_i(\mathcal{N}-1)-z_i^{q+1})\\
&+u_i(\mathcal{N}-1)^TR_iu_i(\mathcal{N}-1)+[u_i(\mathcal{N}-1)-u_i^q(\mathcal{N}-1)]^TH_i[u_i(\mathcal{N}-1)-u_i^q(\mathcal{N}-1)].
\end{aligned}
\end{equation}
Substituting (\ref{e:systemoriginal}) into (\ref{e:Jn1}) and taking the gradient with respect to $u_i(\mathcal{N}-1)$, one obtains
\begin{equation}\label{e:Jn1gradient}
\begin{aligned}
\nabla J_i(\mathcal{N}&-1)=2B_i^TQ_{i\mathcal{N}}\left[A_ix_i(\mathcal{N}-1)+B_iu_i(\mathcal{N}-1)-z_i^{q+1}\right]+2R_iu_i(\mathcal{N}-1)+2H_i[u_i(\mathcal{N}-1)-u_i^q(\mathcal{N}-1)].
\end{aligned}
\end{equation}
Then, the KKT condition of (\ref{e:Jn1gradient}) can be derived, as
\begin{equation}\label{e:Jn1KKT}
\begin{aligned}
u_i^\star(\mathcal{N}-1)=&-(B_i^TQ_{i\mathcal{N}}B_i+R_i+H_i)^{-1}\left[B_i^TQ_{i\mathcal{N}}A_ix_i(\mathcal{N}-1)-B_i^TQ_{i\mathcal{N}}z_i^{q+1}-H_iu_i^q(\mathcal{N}-1)\right]\\
=&-U_i(\mathcal{N}-1)\left[V_i(\mathcal{N}-1)x_i(\mathcal{N}-1)-W_i(\mathcal{N}-1)\right].
\end{aligned}
\end{equation}
Obviously, the unique solution $u_i^\star(\mathcal{N}-1)$ presented by (\ref{e:Jn1KKT}) leads to the minimum cost $J^\star_i(\mathcal{N}-1)$ since $B_i^TQ_{i\mathcal{N}}B_i+R_i+H_i>0$. Then, substituting (\ref{e:Jn1KKT}) into (\ref{e:Jn1}), one can get the minimum cost as
\begin{equation}\label{e:Jn1min}
\begin{aligned}
J_i^\star(\mathcal{N}-1)=&x_i^T(\mathcal{N}-1)S_{i1}(\mathcal{N}-1)x_i(\mathcal{N}-1)+S_{i2}(\mathcal{N}-1)x_i(\mathcal{N}-1)+S_{i3}(\mathcal{N}-1).
\end{aligned}
\end{equation}
Therefore, (\ref{e:controller})-(\ref{e:step1iteration}) are satisfied for $k=\mathcal{N}-1$. Now, assume that (\ref{e:controller})-(\ref{e:step1iteration}) are correct for $k=\mathcal{M}$, i.e.,
\begin{equation}\label{e:optM}
\begin{aligned}
u_i^\star(\mathcal{M})=&-U_i(\mathcal{M})\left[V_i(\mathcal{M})x_i(\mathcal{M})-W_i(\mathcal{M})\right],\\
L_i^\star(\mathcal{M})=&x_i^T(\mathcal{M})S_{i1}(\mathcal{M})x_i(\mathcal{M})+S_{i2}(\mathcal{M})x_i(\mathcal{M})+S_{i3}(\mathcal{M}).
\end{aligned}
\end{equation}
and $S_{i1}(\mathcal{M})$ is positive semi-definite. From the optimization principle, again, it follows that the optimal control input $u_i^\star(\mathcal{M}-1)$ must minimum $J_i(\mathcal{M}-1)$, where
\begin{equation}\label{e:JnM}
\begin{aligned}
J_i(\mathcal{M}-1)=&L_i^\star(\mathcal{M})+(x_i(\mathcal{M}-1)-z_i^{q+1})^TQ_{i}(x_i(\mathcal{M}-1)-z_i^{q+1})+u_i(\mathcal{M}-1)^TR_iu_i(\mathcal{M}-1)\\
&+[u_i(\mathcal{M}-1)-u_i^q(\mathcal{M}-1)]^TH_i[u_i(\mathcal{M}-1)-u_i^q(\mathcal{M}-1)].
\end{aligned}
\end{equation}
Substituting (\ref{e:systemoriginal}) into (\ref{e:JnM}) and taking the gradient with respect to $u_i(\mathcal{M}-1)$, one obtains
\begin{equation}\label{e:JnMgradient}
\begin{aligned}
\nabla J_i(\mathcal{M}-1)=&2B_i^TS_{i1}(\mathcal{M})\left[A_ix_i(\mathcal{M}-1)+B_iu_i(\mathcal{M}-1)\right]+B_i^TS_{i2}^T(\mathcal{M})+2(R_i+H_i)u_i(\mathcal{M}-1)-2H_iu_i^q(\mathcal{M}-1).
\end{aligned}
\end{equation}
Then, the KKT condition of (\ref{e:JnMgradient}) can be obtained, as
\begin{equation}\label{e:JnMKKT}
\begin{aligned}
u_i^\star(\mathcal{M}-1)=&-(B_i^TS_{i1}(\mathcal{M})B_i+R_i+H_i)^{-1}[B_i^TS_{i1}(\mathcal{M})A_ix_i(\mathcal{N}-1)+\frac{1}{2}B_i^TS_{i2}^T(\mathcal{M})-H_iu_i^q(\mathcal{M}-1)]\\
=&-U_i(\mathcal{M}-1)\left[V_i(\mathcal{M}-1)x_i(\mathcal{M}-1)-W_i(\mathcal{M}-1)\right].
\end{aligned}
\end{equation}
Obviously, the unique solution $u_i^\star(\mathcal{M}-1)$ presented by (\ref{e:JnMKKT}) leads to the minimum cost $J^\star_i(\mathcal{M}-1)$ since $B_i^TS_{i1}(\mathcal{M})B_i+R_i+H_i>0$. Then substituting (\ref{e:JnMKKT}) into (\ref{e:JnM}), one can get the minimum cost as
\begin{equation}\label{e:JnMmin}
\begin{aligned}
J_i^\star(\mathcal{M}-1)&=x_i^T(\mathcal{M}-1)S_{i1}(\mathcal{M}-1)x_i(\mathcal{M}-1)+S_{i2}(\mathcal{M}-1)x_i(\mathcal{M}-1)+S_{i3}(\mathcal{M}-1),
\end{aligned}
\end{equation}
which indicates that (\ref{e:controller})-(\ref{e:step1iteration}) are satisfied for $k=\mathcal{M}-1$. In conclusion, the control input sequence $u_i^\star(k),k=0,1,...,\mathcal{N}-1$, minimizes the cost functional $L_\rho(U,Z^{q+1},\Lambda^q)$ in (\ref{e:admmU}) subject to (\ref{e:systemoriginal}), and the optimal objective value can be calculated by (\ref{e:optimalvalue}).
\end{proof}

With the results presented above, a distributed algorithm is established for the linear quadratic synchronization control problem.

\begin{algorithm}[htb]
\caption{Distributed Linear Quadratic Synchronization Control Design Algorithm}
\label{a:disLQR}
\begin{algorithmic}[1]
\Require Initialize $q=0,~\rho>0,~x_i^0(k)\in \mathbb{R}^n,~u_i^0(k)\in \mathbb{R}^m,~z^0_i\in \mathbb{R}^n,~\lambda_i^0, G_i>0,~H_i>({L_\delta}+\frac{L_\delta^2}{2\sigma_{min}\{R_i\}})I_m$, for all $i\in\{1,2,\cdots,N\},~k=0,1,2,\cdots,\mathcal{N}$. Set the stop condition $N_q>0$. For subsystem $i\in\{1,2,\cdots,N\}$ do in parallel:
\Repeat
%\State Send $z_i(k),~\lambda_i^q$ to all the neighbors of subsystem $i$;
\State Solve (\ref{e:step1}) with communication to obtain the equilibrium point $z_{ie}$;
\State Update the synchronization state $z_i^{q+1}=z_{ie}$;
\For{$k=\mathcal{N}-1$ to $0$}
\State Compute the control input $u_i(k)$ from (\ref{e:controller});
\EndFor
\State Update the state $x_i^{q+1}(k),~k=1,2,\cdots,\mathcal{N}$ from (\ref{e:systemoriginal});
\State Update the Lagrangian multiplier $\lambda_i^{q+1}$ according to (\ref{e:admmL});
%\If{$q\ge1$}
%\State $\epsilon_z=\|z_i^{q+1}-z_i^{q}\|$
%\EndIf
\State Set $q=q+1$;
%Let $u_i \triangleq col\{u_i(1),u_i(2),\cdots,u_i(\mathcal{N})\}$. Initialize $u_i^0\in \mathbb{U}_i, z^0_i\in \mathbb{Z}_i$ and $\lambda^0$.  For $q=0,1,...$ until convergence
%\begin{subequations}\label{se:admm}
%\begin{align}
%\label{e:admmZ}
%z_i^{q+1}&=\arg \min \limits_{z_i}\{L_\rho(u_i^{q},z_i,x_0,\Lambda^q)+\frac{1}{2}(z_i-z_i^q)^TG_i(z_i-z_i^q)\},\\
%\label{e:admmU}
%u_i^{q+1}&=\arg \min \limits_{u_i} \{L_\rho(u_i,z_i^{q+1},x_0,\Lambda^q)+\frac{1}{2}(u_i-u_i^q)^T(I_\mathcal{N}\otimes H_i)(u_i-u_i^q)\},\\
%\label{e:admml}
%\lambda^{q+1}_i&=\lambda^q_i+\rho z_i^{q+1}.
%%\label{e:admmlambda}
%%\Lambda^{q+1}&=\Lambda^q+\rho(\mathcal{L}\otimes I_n)Z^{q+1}.
%\end{align}
%\end{subequations}
\Until{$q>N_q$}
\end{algorithmic}
\end{algorithm}

\section{Examples with Simulations}\label{s:Simulations}
\subsection{A Homogeneous System}\label{s:HomogeneousSim}
A scenario of three homogeneous agents is considered first. The edge set of the communication topology is $\{(1,2),(1,3)\}$ and the corresponding Laplacian matrix is
$$\mathcal{L}=
\begin{bmatrix}
2 &-1 &-1\\
-1 &1 &0\\
-1 &0 &1
\end{bmatrix}.
$$
%\begin{figure}[!htb]
%  \centering
%  \includegraphics[scale=0.6]{topology.eps}
%  \caption{The communication topology.}
%  \label{f:topology}
%\end{figure}
Let the agents in (\ref{e:systemoriginal}) be neutrally unstable systems with
\begin{equation*}\label{e:neutrallyunstable}
A_i=
\begin{bmatrix}
1 &1\\
0 &1
\end{bmatrix},~
B_i=
\begin{bmatrix}
0\\
1
\end{bmatrix},~i\in\{1,2,3\}.~
\end{equation*}
The weighted matrices in cost functional (\ref{e:cost}) are set as $Q_i=I_2,~Q_{i\mathcal{N}}=I_2,~R_i=1,~i\in\{1,2,3\}$, and $\mathcal{N}=40$. Choose the parameters in Algorithm \ref{a:disLQR} as $G_i=I_2,~H_i=100,~i\in\{1,2,3\}$ and $\rho=1$. The initial condition is taken as $x_1(0)=[0,0]^T,~x_2(0)=[10,-4]^T,~x_3(0)=[-20,10]^T$, $u_i^0(k)=0,~z_i^0=[0,0]^T,~i\in\{1,2,3\}$. For comparison, the static state-feedback (SSF) method proposed in \cite{LZK2011} is also simulated to verify the effectiveness of Algorithm \ref{a:disLQR} derived in this paper.

Define the trajectories of synchronization error and control cost as $e(k)=(\mathcal{L}\otimes I_n)\times col\{x_1(k),x_2(k),x_3(k)\}$ and $\|u(k)\|=\|col\{u_1(k),u_2(k),u_3(k)\}\|$, respectively. The response trajectories generated by Algorithm \ref{a:disLQR} and the SSF method are depicted in  Fig. \ref{f:neutrallyunstable}, from which it can be seen that the controller designed by Algorithm \ref{a:disLQR} achieves synchronization faster and requires less control energy.
\begin{figure}[!htbp]
\centering
\begin{minipage}[c]{0.45\textwidth}
\centering
\includegraphics[scale=0.5]{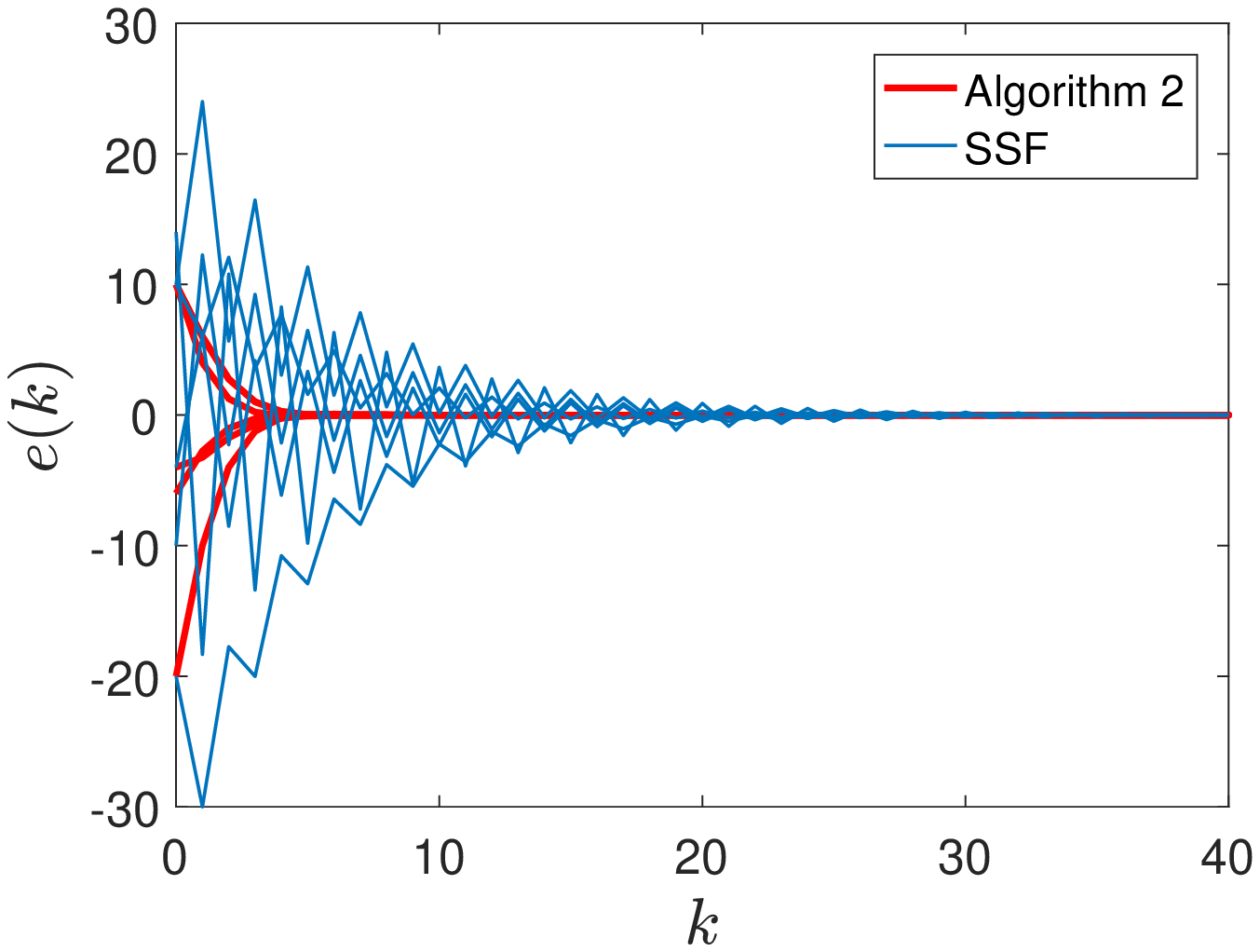}
\end{minipage}
\begin{minipage}[c]{0.45\textwidth}
\centering
\includegraphics[scale=0.5]{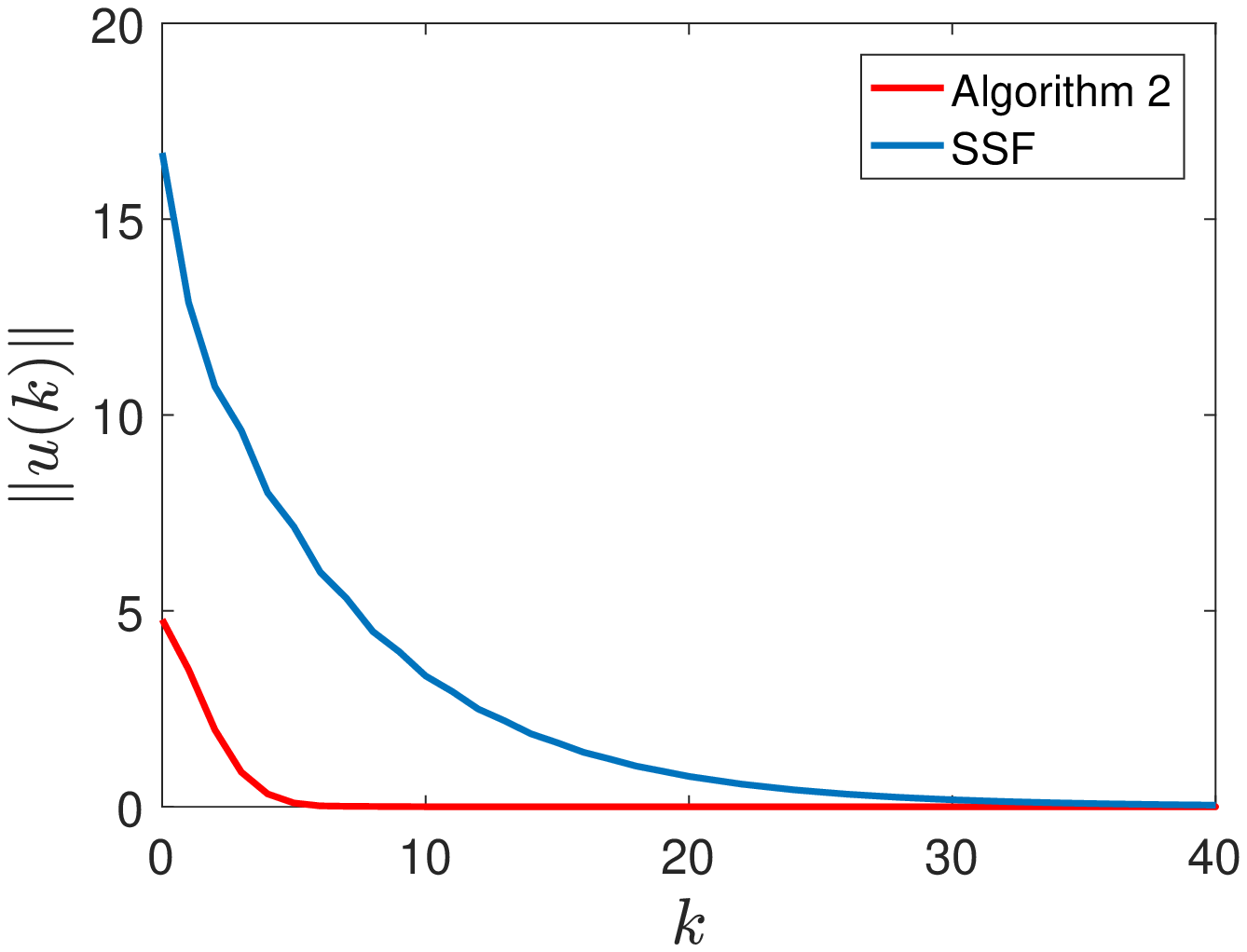}
\end{minipage}
\caption{The curves of $e(k)$ and $\|u(k)\|$ with neutrally unstable agents}
\label{f:neutrallyunstable}
\end{figure}
%\begin{enumerate}[]
%\item Agents with stable dynamics
%\item Agents with stable dynamics
%\item Agents with stable dynamics
%\item Agents with stable dynamics
%\end{enumerate}

In addition, more scenarios such as stable, unstable and neutrally stable dynamics are studied to give a more comprehensive view of the advantages of Algorithm \ref{a:disLQR}. A quantitative comparison is displayed in Table \ref{t:comparison}. Here, the relative cost functional is denoted as $$J=e^T(\mathcal{N})Q_{\mathcal{N}}e(\mathcal{N})+\sum_{k=0}^{\mathcal{N}-1}\left[e(k)^TQe(k)+u(k)^TRu(k)\right],$$
where $R=diag\{R_1,R_2,R_3\},~u(k)=col\{u_1(k),u_2(k),u_3(k)\}$. In both scenarios, Algorithm \ref{a:disLQR} achieves a smaller relative cost and, the more unstable the dynamics are, the better effect the new technique has. From the unstable scenario, it is interesting to see that the Algorithm \ref{a:disLQR} always has a stable solution even if the unstable eigenvalues are far from the unit circle.

%\begin{onecolumn}
\begin{table*}[!htbp]
\centering
\caption{Quantitative Comparison}\label{t:comparison}
\begin{tabular}{l|cc|c|cl}
\hline
\hline
Scenario & $A_i$ & $B_i$ & Method & Relative Cost Functional & Synchronization State\\
\hline
\multirow{2}{*}{Stable}
& \multirow{2}{*}{$\begin{bmatrix} 0.2 & 1\\ 0 & 0.2 \end{bmatrix}$}
& \multirow{2}{*}{$\begin{bmatrix} 0 \\ 1 \end{bmatrix}$}
& ADMM  & 814.93     & $[0.004,~0.02]^T$ \\
& & & SSF   & 1416.82    & $[0,~0]^T$ \\
\hline
\multirow{2}{*}{Neutrally Stable}
& \multirow{2}{*}{$\begin{bmatrix} 0.2 & 1\\ 0 & 1 \end{bmatrix}$}
& \multirow{2}{*}{$\begin{bmatrix} 0 \\ 1 \end{bmatrix}$}
& ADMM  & 907.71    & $[0.33,~0.28]^T$ \\
& & & SSF   & 2219.66   & $[-4.12,~-3.29]^T$ \\
\hline
\multirow{2}{*}{Neutrally Unstable}
& \multirow{2}{*}{$\begin{bmatrix} 1 & 1\\ 0 & 1 \end{bmatrix}$}
& \multirow{2}{*}{$\begin{bmatrix} 0 \\ 1 \end{bmatrix}$}
& ADMM  & 1039.36   & $[-1.69,~0.02]^T$ \\
& & & SSF   & 6924.04   & $[-\infty,~-3.29]^T$ \\
\hline
\multirow{2}{*}{Unstable1}
& \multirow{2}{*}{$\begin{bmatrix} 1.2 & 1\\ 0 & 1 \end{bmatrix}$}
& \multirow{2}{*}{$\begin{bmatrix} 0 \\ 1 \end{bmatrix}$}
& ADMM  & 1.38e3    & $[-2.26,~0.47]^T$ \\
& & & SSF   & 2.25e4    & $[-\infty,~-\infty]^T$ \\
\hline
\multirow{2}{*}{Unstable2}
& \multirow{2}{*}{$\begin{bmatrix} 2 & 1\\ 0 & 1 \end{bmatrix}$}
& \multirow{2}{*}{$\begin{bmatrix} 0 \\ 1 \end{bmatrix}$}
& ADMM  & 8.74e3    & $[-4.28,~4.26]^T$ \\
& & & SSF   & NAN       & NAN \\
\hline
\end{tabular}
\end{table*}
%\end{onecolumn}

%\begin{figure}[!hbp]
%\centering
%\begin{minipage}[c]{0.5\textwidth}
%\centering
%\includegraphics[height=4.5cm,width=7.5cm]{stable_xe.eps}
%\end{minipage}%
%\begin{minipage}[c]{0.5\textwidth}
%\centering
%\includegraphics[height=4.5cm,width=7.5cm]{stable_nu.eps}
%\end{minipage}
%\caption{The curves of synchronization error and control cost with stable agents}
%\end{figure}
%
%\begin{figure}[!hbp]
%\centering
%\begin{minipage}[c]{0.5\textwidth}
%\centering
%\includegraphics[height=4.5cm,width=7.5cm]{neustable_xe.eps}
%\end{minipage}%
%\begin{minipage}[c]{0.5\textwidth}
%\centering
%\includegraphics[height=4.5cm,width=7.5cm]{neustable_nu.eps}
%\end{minipage}
%\caption{The curves of synchronization error and control cost with neutrally stable agents}
%\end{figure}
%
%
%
%\begin{figure}[!hbp]
%\centering
%\begin{minipage}[c]{0.5\textwidth}
%\centering
%\includegraphics[height=4.5cm,width=7.5cm]{unstable_xe.eps}
%\end{minipage}%
%\begin{minipage}[c]{0.5\textwidth}
%\centering
%\includegraphics[height=4.5cm,width=7.5cm]{unstable_nu.eps}
%\end{minipage}
%\caption{The curves of synchronization error and control cost with unstable agents}
%\end{figure}

%\begin{multicols}{2}
\subsection{A Heterogeneous System}
Now, it is to demonstrate the effectiveness of Algorithm \ref{a:disLQR} in the heterogeneous scenario. Consider a network of agents described by (\ref{e:systemoriginal}) with
\begin{equation*}
\begin{aligned}
&A_1=
\begin{bmatrix}
1.2 &1 &2\\
0 &2.4 &2\\
2 &0   &1.5
\end{bmatrix},~
A_2=
\begin{bmatrix}
0 &1.3 &-0.7\\
0.5 &0.85 &0.85\\
0.5 &-0.65   &1.35
\end{bmatrix},
A_3=
\begin{bmatrix}
0.3 &1 &0\\
0 &1.2 &1\\
0 &0   &0.4
\end{bmatrix},
\\&B_1=
\begin{bmatrix}
0 & 1 & 1\\
2 & 0 & -1
\end{bmatrix}^T,
B_2=
\begin{bmatrix}
0 & 0 & 1\\
0 & 2 & 0
\end{bmatrix}^T,
B_3=
\begin{bmatrix}
0 & 1 & 0\\
-1 & 0 & -2
\end{bmatrix}^T.
\end{aligned}
\end{equation*}

\begin{figure}[!htb]
\centering
\begin{minipage}[c]{0.45\textwidth}
\centering
\includegraphics[scale=0.5]{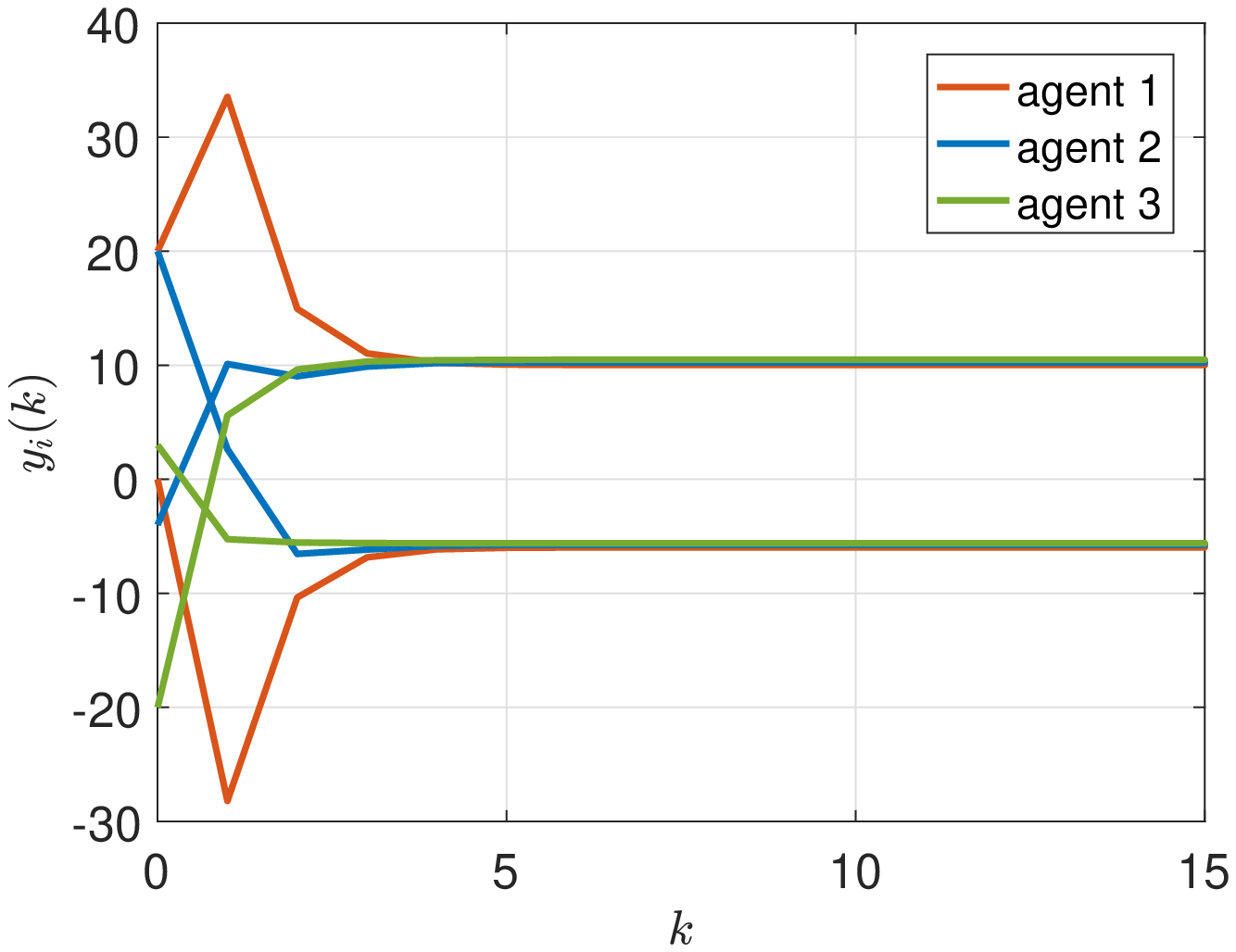}
\end{minipage}
\begin{minipage}[c]{0.45\textwidth}
\centering
\includegraphics[scale=0.5]{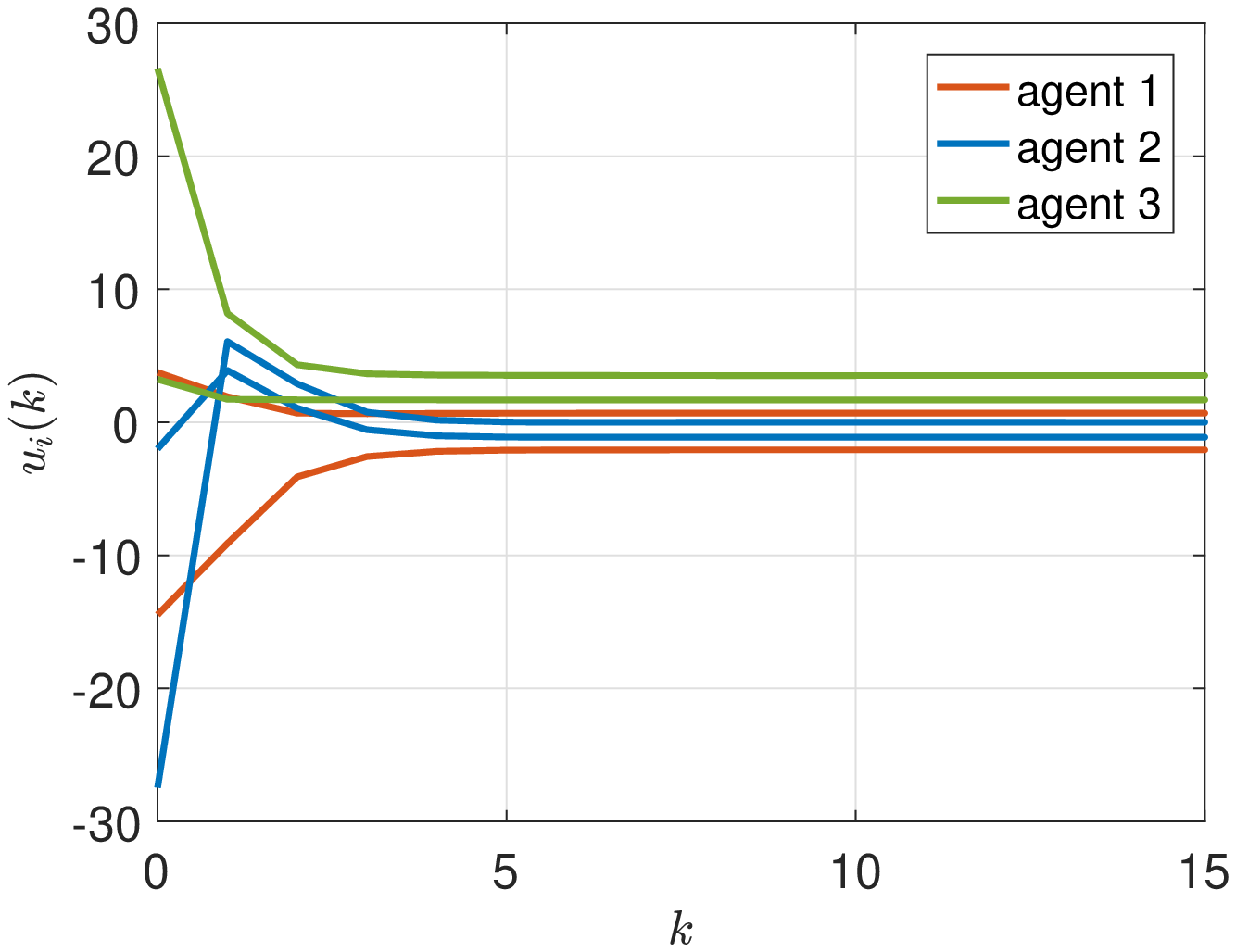}
\end{minipage}
\caption{The curves of output and input signals for each agent}
\label{f:heterogeneous}
\end{figure}
Assume that the communication topology is given, the same as that in Subsection \ref{s:HomogeneousSim}. The weighted matrices in cost functional (\ref{e:cost}) are set as $Q_1=diag\{0,8,13\},~Q_{1\mathcal{N}}=diag\{0,8,1\},~Q_2=diag\{0,3,5\},~Q_{1\mathcal{N}}=diag\{0,5,1\},~Q_3=diag\{0,4,15\},~Q_{3\mathcal{N}}=diag\{0,12,5\},~R_i=I_2,~i\in\{1,2,3\}$, and $\mathcal{N}=50$. In this scenario, the weighted matrices $Q_i$ and $Q_{1\mathcal{N}}$ are selected as positive semi-definite matrices, i.e., $y_i=[0,1,1]x_i,~i\in\{1,2,3\}$, to demonstrate the output synchronization ability of the proposed algorithm. Choose the parameters in Algorithm \ref{a:disLQR} as $G_i=I_3,~H_i=1e3\times I_2,~i\in\{1,2,3\}$, and $\rho=1$. The initial condition is taken as $x_1(0)=[-5,20,0]^T,~x_2(0)=[1,-4,20]^T,~x_3(0)=[-2,-20,3]^T$, $u_i^0(k)=[0,0]^T,~z_i^0=[0,0,0]^T,~i\in\{1,2,3\}$. The trajectories of the last two components of the states and the control inputs are shown in Fig. \ref{f:heterogeneous}, which indicates that the outputs of the agents synchronize rapidly and the control inputs converge (to different values) to maintain the synchronization.

\section{Conclusions}\label{s:conclusion}
The distributed optimal synchronization problem with linear quadratic cost is solved in this paper for multi-agent systems with a undirected communication topology. The optimal synchronization problem is formulated as a distributed optimization problem with a linear quadratic cost functional that integrates the energies of the synchronization error signal and of the input signal. By the application of a modified ADMM technique, the optimal synchronization control problem is separated into the synchronization step and the optimal control step. These two subproblems are then solved by distributed numerical algorithms based on the Lyapunov method and dynamic programming. The performances of the proposed design are is demonstrated by numerical examples for both homogenous and heterogenous linear multi-agent systems with either stable or unstable dynamics.

\appendices
\section{Proof of Theorem \ref{t:convergence}}
Before proceeding to the convergence analysis, a useful lemma is first introduced.
\begin{lemma} \cite{DroriY2015}\label{l:convexinequality}
For any convex function $f$ on $\mathbb{R}^m$, which is continuously differentiable with gradient $\nabla f$ satisfying the Lipschitz continuous condition
\begin{equation}\label{Lipschitzf}
\|\nabla f(x)-\nabla f(y)\|\le L_f\|x-y\|,~~\forall x,y\in \mathbb{R}^m,
\end{equation}
one has
\begin{equation}\label{e:convexeq}
\begin{aligned}
f(x)\le& f(y)+\nabla f(z)^T(x-y)+\frac{L_f}{2}\|x-z\|^2,\forall x,y,z\in \mathbb{R}^m.
\end{aligned}
\end{equation}
\end{lemma}

%\begin{lemma}\label{l:optimalsaddle}
%For any convex function $f$ on $\mathbb{R}^m$ which is continuously differentiable with gradient $\nabla f$ which satisfies the Lipschitz continuous condition as follows
%\begin{equation}\label{Lipschitzf}
%\|\nabla f(x)-\nabla f(y)\|\le L_f\|x-y\|,~~~~~\forall x,y\in \mathbb{R}^m,
%\end{equation}
%we have
%\begin{equation}\label{e:convexeq}
%f(x)\le f(y)+\nabla f(z)^T(x-y)+\frac{L_f}{2}\|x-z\|^2,~~~~~\forall x,y,z\in \mathbb{R}^m.
%\end{equation}
%\end{lemma}

Next, the proof of Theorem \ref{t:convergence} is presented.
\begin{proof}
%See Appendix A.
By substituting (\ref{e:systemoriginal}) into (\ref{e:cost}), one has
\begin{equation}\label{e:costpar}
J(U,Z)=J_1(U,Z)+J_2(U),
\end{equation}
where
\begin{equation}\label{e:costJ1J2}
\begin{aligned}
J_1&(U,Z)=\sum_{i=1}^{N}\left\{\left(A_i^\mathcal{N}x_i(0)+\sum_{j=0}^{\mathcal{N}-1}A_i^{\mathcal{N}-1-j}B_iu_i(j)\right.-z_i\right)^TQ_{i\mathcal{N}}\left(A_i^\mathcal{N}x_i(0)+\sum_{j=0}^{\mathcal{N}-1}A_i^{\mathcal{N}-1-j}B_iu_i(j)-z_i\right)\\
&+\sum_{k=0}^{\mathcal{N}-1}\left[\left(A_i^kx_i(0)+\sum_{j=0}^{k-1}A_i^{k-1-j}B_iu_i(j)-z_i\right)^T\left.Q_{i}\left(A_i^kx_i(0)+\sum_{j=0}^{k-1}A_i^{k-1-j}B_iu_i(j)-z_i\right)\right]\right\},\\
J_2&(U)=\sum_{i=1}^{N}\sum_{k=0}^{\mathcal{N}-1}u_i^T(k)R_iu_i(k).
\end{aligned}
\end{equation}

It is easy to see that $J_1(U,Z)$ is convex with respect to $u_i,z_i$ and $J_2(U)$ is strongly convex with respect to $u_i$ since $Q_{i\mathcal{N}}\geq0,Q_i\geq0,R_i>0$. Then, the gradient of $J_1(U,Z)$ %respect to $u_i$ and $z_i$ can be obtained as (\ref{gradientui}) and (\ref{gradientzi}),
can be obtained as
\begin{equation}\label{gradientui}
\nabla_{u_i}J_1=2
\begin{bmatrix}
\begin{aligned}
&\left(A_i^{\mathcal{N}-1}B_i\right)^TQ_{i\mathcal{N}}\left(A_i^\mathcal{N}x_i(0)+\sum_{j=0}^{\mathcal{N}-1}A_i^{\mathcal{N}-1-j}B_iu_i(j)-z_i\right)\\
&+\sum_{k=1}^{\mathcal{N}-1}\left(A_i^{k-1}B_i\right)^TQ_i\left(A_i^kx_i(0)+\sum_{j=0}^{k-1}A_i^{k-1-j}B_iu_i(j)-z_i\right)
\end{aligned}\\[10ex]
\begin{aligned}
\left(A_i^{\mathcal{N}-2}B_i\right)^TQ_{i\mathcal{N}}\left(A_i^\mathcal{N}x_i(0)+\sum_{j=0}^{\mathcal{N}-1}A_i^{\mathcal{N}-1-j}B_iu_i(j)-z_i\right)\\
+\sum_{k=2}^{\mathcal{N}-1}\left(A_i^{k-2}B_i\right)^TQ_i\left(A_i^kx_i(0)+\sum_{j=0}^{k-1}A_i^{k-1-j}B_iu_i(j)-z_i\right)
\end{aligned}\\[5ex]
\vdots\\
B_i^TQ_{i\mathcal{N}}\left(A_i^\mathcal{N}x_i(0)+\sum_{j=0}^{\mathcal{N}-1}A_i^{\mathcal{N}-1-j}B_iu_i(j)-z_i\right)
\end{bmatrix},
\end{equation}
\begin{equation}\label{gradientzi}
\begin{aligned}
\nabla_{z_i}J_1=&-2Q_{i\mathcal{N}}\left(A_i^\mathcal{N}x_i(0)+\sum_{j=0}^{\mathcal{N}-1}A_i^{\mathcal{N}-1-j}B_iu_i(j)-z_i\right)-Q_i\left(x_i(0)-z_i\right)\\
&-2\sum_{k=1}^{\mathcal{N}-1}Q_i\left(A_i^kx_i(0)+\sum_{j=0}^{k-1}A_i^{k-1-j}B_iu_i(j)-z_i\right),
\end{aligned}
\end{equation}
which can be rewritten in a compact form as
\begin{equation}\label{gradient}
\setcounter{equation}{33}
\begin{aligned}
\nabla J_1=
\begin{bmatrix}
\nabla_U J_1\\
\nabla_Z J_1
\end{bmatrix}=&L_0(A_i,B_i,Q_i,Q_{i\mathcal{N}})
\begin{bmatrix}
x_1(0)\\
\vdots\\
x_N(0)
\end{bmatrix}
+L_\Delta(A_i,B_i,Q_i,Q_{i\mathcal{N}})
\begin{bmatrix}
U\\
Z
\end{bmatrix}.
\end{aligned}
\end{equation}

Therefore, the cost functional $J_1(U,Z)$ satisfies
\begin{equation}\label{e:Lipschitz}
\begin{aligned}
\|\nabla J_1(U_1,Z_1)-\nabla J_1(U_2,Z_2)\|&=\left\| L_\Delta(A_i,B_i,Q_i,Q_{i\mathcal{N}},R_i)
\begin{bmatrix}
U_1-U_2\\
Z_1-Z_2
\end{bmatrix}\right\|\\
&\le\|L_\Delta(A_i,B_i,Q_i,Q_{i\mathcal{N}},R_i)\|
\left\|
\begin{bmatrix}
U_1\\
Z_1
\end{bmatrix}-
\begin{bmatrix}
U_2\\
Z_2
\end{bmatrix}\right\|\\
&\le L_\delta
\left\|
\begin{bmatrix}
U_1\\
Z_1
\end{bmatrix}-
\begin{bmatrix}
U_2\\
Z_2
\end{bmatrix}\right\|,
~~~\forall
\begin{bmatrix}
U_1\\
Z_1
\end{bmatrix},~
\begin{bmatrix}
U_2\\
Z_2
\end{bmatrix}.
\end{aligned}
\end{equation}
where $L_\delta$ is a Lipschitz constant for $\nabla J_1(U,Z)$. In the following, the convergence of Algorithm \ref{a:ADMM} is proved.
%For convenience of the analysis, we rewrite Algorithm \ref{a:ADMM} in a compact form:
%\begin{subequations}\label{se:admmC}
%\begin{align}
%\label{e:admmZ}
%Z^{q+1}&=\arg \min \limits_{Z}\{L_\rho(U^{q},Z,x_0,\Lambda^q)+\frac{1}{2}(Z-Z^q)^TG(Z-Z^q)\},\\
%\label{e:admmU}
%U^{q+1}&=\arg \min\limits_{U}\{L_\rho(U,Z^{q+1},x_0,\lambda^q)+\frac{1}{2}(U-U^q)^TH(U-U^q)\},\\
%\label{e:admmL}
%\Lambda^{q+1}&=\Lambda^q+\rho Z^{q+1},
%%\label{e:admmlambda}
%%\Lambda^{q+1}&=\Lambda^q+\rho(\mathcal{L}\otimes I_n)Z^{q+1}.
%\end{align}
%\end{subequations}
%where $G\triangleq diag\{G_1,G_2,\cdots,G_N\}$ and $H\triangleq diag\{I_\mathcal{N}\otimes H_1,I_\mathcal{N}\otimes H_2,\cdots,I_\mathcal{N}\otimes H_N\}$.
According to (\ref{e:Lagrangian}), the augmented Lagrangian can be written as
\begin{equation}\label{e:Lagrangianp}
\begin{aligned}
L_\rho(U,Z,\Lambda)=&J(U,Z)+\Lambda^T(\mathcal{L}\otimes I_n)Z+\frac{\rho}{2}Z^T(\mathcal{L}\otimes I_n)Z,
\end{aligned}
\end{equation}
By the optimality condition \cite{FacchineiF2003}, the optimal solution of subproblems (\ref{e:admmZ}) and (\ref{e:admmU}) satisfies
\begin{equation}\label{e:saddlepointZ}
\begin{aligned}
&(Z-Z^{q+1})^T\left[\nabla_ZJ_1(U^{q},Z^{q+1})+G(Z^{q+1}-Z^q)+\rho(\mathcal{L}\otimes I_n)Z^{q+1}+(\mathcal{L}\otimes I_n)\Lambda^{q}\right]\ge 0,~~~\forall Z\in \mathbb{R}^n,
\end{aligned}
\end{equation}
and
\begin{equation}\label{e:saddlepointU}
\begin{aligned}
&(U-U^{q+1})^T\left[\nabla_UJ_1(U^{q+1},Z^{q+1})+2\bar RU^{q+1}+H(U^{q+1}-U^q)\right]\ge 0,~~~\forall U\in \mathbb{R}^m,
\end{aligned}
\end{equation}
where $\bar R=diag\{I_\mathcal{N}\otimes R_1,I_\mathcal{N}\otimes R_2,\cdots,I_\mathcal{N}\otimes R_N\}$. By the Lipschitz continuity and Lemma \ref{l:convexinequality}, one can get
\begin{equation}\label{e:gradienteq1}
\begin{aligned}
&(Z-Z^{q+1})^T\nabla_Z J_1(U^q,Z^{q+1})+(U-U^{q+1})^T\nabla_U J_1(U^{q+1},Z^{q+1})\\
=&(Z-Z^{q+1})^T\nabla_Z J_1(U^q,Z^{q+1})+(U-U^{q+1})\nabla_U J_1(U^q,Z^{q+1})\\
&+(U-U^{q+1})^T\left[\nabla_UJ_1(U^{q+1},Z^{q+1})-\nabla_U J_1(U^q,Z^{q+1})\right]\\
\le&(Z-Z^{q+1})^T\nabla_ZJ_1(U^q,Z^{q+1})+(U-U^{q+1})\nabla_UJ_1(U^q,Z^{q+1})+L_\delta\|U-U^{q+1}\|\|U^{q+1}-U^q\|\\
=&(Z-Z^{q+1})^T\nabla_ZJ_1(U^q,Z^{q+1})+(U-U^q)^T\nabla_UJ_1(U^q,Z^{q+1})+(U^q-U^{q+1})^T\nabla_UJ_1(U^q,Z^{q+1})\\
&+L_\delta\|U-U^{q+1}\|\|U^{q+1}-U^q\|\\
\le&J_1(U,Z)-J_1(U^{q},Z^{q+1})+(U^q-U^{q+1})^T\nabla_UJ_1(U^q,Z^{q+1})+L_\delta\|U-U^{q+1}\|\|U^{q+1}-U^q\|,
\end{aligned}
\end{equation}
and
\begin{equation}\label{e:gradienteq2}
\begin{aligned}
&J_1(U,Z)-J_1(U^{q},Z^{q+1})+(U^q-U^{q+1})^T\nabla_UJ_1(U^q,Z^{q+1})+L_\delta\|U-U^{q+1}\|\|U^{q+1}-U^q\|\\
\le&J_1(U,Z)-J_1(U^{q},Z^{q+1})+J_1(U^{q},Z^{q+1})-J_1(U^{q+1},Z^{q+1})+\frac{L_\delta}{2}\|U^q-U^{q+1}\|^2+L_\delta\|U-U^{q+1}\|\|U^{q+1}-U^q\|\\
\le&J_1(U,Z)-J_1(U^{q+1},Z^{q+1})+\left(\frac{L_\delta}{2}+\frac{L_\delta^2}{4\sigma_{min}\{R_i\}}\right)\|U^q-U^{q+1}\|^2+\sigma_{min}\{R_i\}\|U-U^{q+1}\|^2,
\end{aligned}
\end{equation}
where $\sigma_{min}\{R_i\}$ denotes the minimum value of the eigenvalues of $R_1,R_2,\cdots,R_N$, and the last two inequalities follow from (\ref{Lipschitzf}) and (\ref{e:convexeq}). Combining (\ref{e:saddlepointZ}) and (\ref{e:saddlepointU}) yields
\begin{equation}\label{e:saddlepointUZ}
\begin{aligned}
0\le& (Z-Z^{q+1})^T\left[\nabla_ZJ_1(U^{q},Z^{q+1})+G(Z^{q+1}-Z^q)+\rho(\mathcal{L}\otimes I_n)Z^{q+1}+(\mathcal{L}\otimes I_n)\Lambda^{q}\right]\\
&+(U-U^{q+1})^T\left[\nabla_UJ_1(U^{q+1},Z^{q+1})+2\bar RU^{q+1}+H(U^{q+1}-U^q)\right]\\
\le&J_1(U,Z)-J_1(U^{q+1},Z^{q+1})+(Z-Z^{q+1})^TG(Z^{q+1}-Z^q)+(U-U^{q+1})^TH(U^{q+1}-U^q)\\
&+\sigma_{min}\{R_i\}\|U-U^{q+1}\|^2+\left(\frac{L_\delta}{2}+\frac{L_\delta^2}{4\sigma_{min}\{R_i\}}\right)\|U^q-U^{q+1}\|^2+U^T\bar RU-U^{q+1}\bar RU^{q+1}\\
&-\sigma_{min}\{R_i\}\|U-U^{q+1}\|^2+(Z-Z^{q+1})^T\left[\rho(\mathcal{L}\otimes I_n)Z^{q+1}+(\mathcal{L}\otimes I_n)\Lambda^{q}\right]\\
=&J(U,Z)-J(U^{q+1},Z^{q+1})+(Z-Z^{q+1})^TG(Z^{q+1}-Z^q)+(U-U^{q+1})^TH(U^{q+1}-U^q)\\
&+\left(\frac{L_\delta}{2}+\frac{L_\delta^2}{4\sigma_{min}\{R_i\}}\right)\|U^q-U^{q+1}\|^2+(Z-Z^{q+1})^T\left[\rho(\mathcal{L}\otimes I_n)Z^{q+1}+(\mathcal{L}\otimes I_n)\Lambda^{q}\right].
\end{aligned}
\end{equation}
It is easy to verify that
\begin{equation}\label{e:UGU}
\begin{aligned}
(Z-Z^{q+1})^TG(Z^{q+1}-Z^q)=&-\frac{1}{2}(Z-Z^{q+1})^TG(Z-Z^{q+1})+\frac{1}{2}(Z-Z^{q})^TG(Z-Z^{q})\\
&-\frac{1}{2}(Z^q-Z^{q+1})^TG(Z^q-Z^{q+1}),
\end{aligned}
\end{equation}
and
\begin{equation}\label{e:ZHZ}
\begin{aligned}
&(U-U^{q+1})^TH(U^{q+1}-U^q)+\left(\frac{L_\delta}{2}+\frac{L_\delta^2}{4\sigma_{min}\{R_i\}}\right)\|U^q-U^{q+1}\|^2\\
\le&-\frac{1}{2}(U-U^{q+1})^TH(U-U^{q+1})+\frac{1}{2}(U-U^{q})^TH(U-U^{q})\\
&-\frac{1}{2}(U^q-U^{q+1})^T\left(H-{L_\delta}I+\frac{L_\delta^2I}{2\sigma_{min}\{R_i\}}\right)(U^q-U^{q+1}).
\end{aligned}
\end{equation}
Then, from (\ref{e:admmL}), it follows that
\begin{equation}\label{e:LrL}
\begin{aligned}
&(Z-Z^{q+1})^T\left[\rho(\mathcal{L}\otimes I_n)Z^{q+1}+(\mathcal{L}\otimes I_n)\Lambda^{q}\right]\\
=&(\Lambda-\Lambda^{q+1})^T\left[\frac{1}{\rho}(\mathcal{L}\otimes I_n)(\Lambda^{q+1}-\Lambda^q)-(\mathcal{L}\otimes I_n)Z^{q+1}\right]+(Z-Z^{q+1})^T(\mathcal{L}\otimes I_n)\Lambda^{q+1}\\
=&-\frac{1}{2}(\Lambda-\Lambda^{q+1})^T\frac{1}{\rho}(\mathcal{L}\otimes I_n)(\Lambda-\Lambda^{q+1})+\frac{1}{2}(\Lambda-\Lambda^{q})^T\frac{1}{\rho}(\mathcal{L}\otimes I_n)(\Lambda-\Lambda^{q})\\
&-\frac{1}{2}(\Lambda^q-\Lambda^{q+1})^T\frac{1}{\rho}(\mathcal{L}\otimes I_n)(\Lambda^q-\Lambda^{q+1})-(\Lambda-\Lambda^{q+1})^T(\mathcal{L}\otimes I_n)Z+(Z-Z^{q+1})^T(\mathcal{L}\otimes I_n)\Lambda .
\end{aligned}
\end{equation}

Substituting (\ref{e:UGU}), (\ref{e:ZHZ}) and (\ref{e:LrL}) into (\ref{e:saddlepointUZ}), gives
%\begin{equation}\label{e:saddlepointUZ2}
%\begin{small}
%\begin{aligned}
%0\le&J(U,Z)-J(U^{q+1},Z^{q+1})-(\Lambda-\Lambda^{q+1})^T(\mathcal{L}\otimes I_n)Z\\
%&+(Z-Z^{q+1})^T(\mathcal{L}\otimes I_n)\Lambda\\
%&+\frac{1}{2}
%\begin{bmatrix}
%U-U^q\\
%Z-Z^q\\
%\Lambda-\Lambda^q
%\end{bmatrix}^T
%\begin{bmatrix}
%H &0 &0\\
%0 &G &0\\
%0 &0 &\frac{1}{\rho}(\mathcal{L}\otimes I_n)
%\end{bmatrix}
%\begin{bmatrix}
%U-U^q\\
%Z-Z^q\\
%\Lambda-\Lambda^q
%\end{bmatrix}\\
%&-\frac{1}{2}
%\begin{bmatrix}
%U-U^{q+1}\\
%Z-Z^{q+1}\\
%\Lambda-\Lambda^{q+1}
%\end{bmatrix}^T
%\begin{bmatrix}
%H &0 &0\\
%0 &G &0\\
%0 &0 &\frac{1}{\rho}(\mathcal{L}\otimes I_n)
%\end{bmatrix}
%\begin{bmatrix}
%U-U^{q+1}\\
%Z-Z^{q+1}\\
%\Lambda-\Lambda^{q+1}
%\end{bmatrix}\\
%&-\frac{1}{2}
%\begin{bmatrix}
%U^q-U^{q+1}\\
%Z^q-Z^{q+1}\\
%\Lambda^q-\Lambda^{q+1}
%\end{bmatrix}^T
%\begin{bmatrix}
%H-{L_\delta I}+\frac{L_\delta^2I}{2\sigma_{min}\{R_i\}} &0 &0\\
%0 &G &0\\
%0 &0 &\frac{1}{\rho}(\mathcal{L}\otimes I_n)
%\end{bmatrix}
%\begin{bmatrix}
%U^q-U^{q+1}\\
%Z^q-Z^{q+1}\\
%\Lambda^q-\Lambda^{q+1}
%\end{bmatrix}.
%\end{aligned}
%\end{small}
%\end{equation}
\begin{equation}\label{e:saddlepointUZ2}
\begin{aligned}
0\le&J(U,Z)-(\Lambda-\Lambda^{q+1})^T(\mathcal{L}\otimes I_n)Z-J(U^{q+1},Z^{q+1})+(Z-Z^{q+1})^T(\mathcal{L}\otimes I_n)\Lambda\\
&+\frac{1}{2}
\begin{bmatrix}
U-U^q\\
Z-Z^q\\
\Lambda-\Lambda^q
\end{bmatrix}^T
M_1
\begin{bmatrix}
U-U^q\\
Z-Z^q\\
\Lambda-\Lambda^q
\end{bmatrix}-\frac{1}{2}
\begin{bmatrix}
U-U^{q+1}\\
Z-Z^{q+1}\\
\Lambda-\Lambda^{q+1}
\end{bmatrix}^T
M_1
\begin{bmatrix}
U-U^{q+1}\\
Z-Z^{q+1}\\
\Lambda-\Lambda^{q+1}
\end{bmatrix}-\frac{1}{2}
\begin{bmatrix}
U^q-U^{q+1}\\
Z^q-Z^{q+1}\\
\Lambda^q-\Lambda^{q+1}
\end{bmatrix}^T
M_2
\begin{bmatrix}
U^q-U^{q+1}\\
Z^q-Z^{q+1}\\
\Lambda^q-\Lambda^{q+1}
\end{bmatrix},
\end{aligned}
\end{equation}
where
\begin{equation}\label{e:THETAM1M2}
\begin{aligned}
M_1&=\begin{bmatrix}
H &0 &0\\
0 &G &0\\
0 &0 &\frac{1}{\rho}(\mathcal{L}\otimes I_n)
\end{bmatrix},M_2=\begin{bmatrix}
H-{L_\delta I}+\frac{L_\delta^2 I}{2\sigma_{min}\{R_i\}} &0 &0\\
0 &G &0\\
0 &0 &\frac{1}{\rho}(\mathcal{L}\otimes I_n)
\end{bmatrix}.
\end{aligned}
\end{equation}
Letting $U=U^\star,Z=Z^\star,\Lambda=\Lambda^\star$, in which the superscript $\star$ represents the optimal solution, and denoting
\begin{equation}\label{e:THETA}
\Theta=
\begin{bmatrix}
U^T &Z^T &\Lambda^T
\end{bmatrix}^T,
\end{equation}
%\begin{equation}\label{e:THETAM1M2}
%\Theta=\begin{bmatrix}
%U\\
%Z\\
%\Lambda
%\end{bmatrix},
%M_1=\begin{bmatrix}
%H &0 &0\\
%0 &G &0\\
%0 &0 &\frac{1}{\rho}(\mathcal{L}\otimes I_n)
%\end{bmatrix},
%M_2=\begin{bmatrix}
%H-{L_\delta I}+\frac{L_\delta^2 I}{2\sigma_{min}\{R_i\}} &0 &0\\
%0 &G &0\\
%0 &0 &\frac{1}{\rho}(\mathcal{L}\otimes I_n)
%\end{bmatrix},
%\end{equation}
one obtains
\begin{equation}\label{e:optimalineq}
\begin{aligned}
&\frac{1}{2}
(\Theta^\star-\Theta^q)^TM_1(\Theta^\star-\Theta^q)-\frac{1}{2}(\Theta^\star-\Theta^{q+1})^TM_1(\Theta^\star-\Theta^{q+1})-\frac{1}{2}(\Theta^q-\Theta^{q+1})^TM_2(\Theta^q-\Theta^{q+1})\\
\ge&J(U^{q+1},Z^{q+1})-J(U^\star,Z^\star)+(\Lambda^\star-\Lambda^{q+1})^T(\mathcal{L}\otimes I_n)Z^\star-(Z^\star-Z^{q+1})^T(\mathcal{L}\otimes I_n)\Lambda^\star\\
\ge&0.
\end{aligned}
\end{equation}
If $G_i>0, H_i>({L_\delta}+\frac{L_\delta^2}{2\sigma_{min}\{R_i\}})I_m $, it can be concluded that $M_1\ge0,M_2\ge0$. From (\ref{e:optimalineq}), one can obtain
\begin{equation}\label{e:convergencesq}
\begin{aligned}
&\frac{1}{2}
(\Theta^\star-\Theta^q)^TM_1(\Theta^\star-\Theta^q)-\frac{1}{2}(\Theta^\star-\Theta^{q+1})^TM_1(\Theta^\star-\Theta^{q+1})\ge\frac{1}{2}(\Theta^q-\Theta^{q+1})^TM_2(\Theta^q-\Theta^{q+1})\ge0,
\end{aligned}
\end{equation}
which means that $\left\{(\Theta^\star-\Theta^q)^TM_1(\Theta^\star-\Theta^q),~q=1,2,\cdots\right\}$ is a decreasing sequence. Then, from $(\Theta^\star-\Theta^q)^TM_1(\Theta^\star-\Theta^q)\ge0$, it follows that the sequence $\left\{(\Theta^\star-\Theta^q)^TM_1(\Theta^\star-\Theta^q),~q=1,2,\cdots\right\}$ is convergent and $\{\Theta^q,~q=1,2,\cdots\}$ is bounded. Therefore, it follows from (\ref{e:convergencesq}) that
\begin{equation}\label{e:convergenceq}
\begin{aligned}
\lim_{q \to +\infty}(\Theta^q-\Theta^{q+1})^TM_2(\Theta^q-\Theta^{q+1})=0,
\end{aligned}
\end{equation}
which implies that $\lim_{q \to +\infty}(U^q-U^{q+1})=0,~\lim_{q \to +\infty}(Z^q-Z^{q+1})=0$ and $\lim_{q \to +\infty}(\mathcal{L}\otimes I_n)(\Lambda^q-\Lambda^{q+1})=0$. Hence, the sequences $(U^q,Z^q)$ and $(U^{q+1},Z^{q+1})$ converge to the same cluster points $(U^\infty,Z^\infty)$. From the first inequality of (\ref{e:saddlepointUZ}) and (\ref{e:LrL}), one gets
\begin{equation}\label{limitsaddlepoint}
\begin{aligned}
&\begin{bmatrix}
U-U^{\infty}\\
Z-Z^{\infty}\\
\Lambda-\Lambda^{\infty}
\end{bmatrix}^T\left\{
\begin{bmatrix}
\nabla_UJ(U^{\infty},Z^{\infty})\\
\nabla_ZJ(U^{\infty},Z^{\infty})\\
0
\end{bmatrix}+
\begin{bmatrix}
0\\
(\mathcal{L}\otimes I_n)\Lambda^{\infty}\\
-(\mathcal{L}\otimes I_n)Z^\infty
\end{bmatrix}\right\}\ge0.
\end{aligned}
\end{equation}
By the ensemble variational inequality \cite{FacchineiF2003}, it consequently follows that $(U^\infty,Z^\infty,\Lambda^\infty)$ is an optimal solution. Therefore, $(U^q,Z^q,\Lambda^q)$ converges to the optimal solution of the distributed linear quadratic synchronization control problem (\ref{e:optimizationproblem}).
\end{proof}

%\section{}

% or


%\appendices
%\section{Proof of Theorem \ref{t:convergence}}

%\section{}

\section*{Acknowledgment}

This work is supported by the National Natural Science Foundation (NNSF) of China under Grants 61673026
and 61528301, and the HongKong Research Grants Council under the GRF Grant CityU 11234916.

\ifCLASSOPTIONcaptionsoff
  \newpage
\fi



%\bibliographystyle{IEEEtran}
%\begin{thebibliography}{1}
%\end{thebibliography}

\bibliographystyle{ieeetr}
\bibliography{distributedLQR}
\end{document}